\documentclass[11pt,twoside]{amsart}
\usepackage[foot]{amsaddr}
\usepackage{amsmath, amsthm, amscd, amsfonts, amssymb, graphicx, color}
\usepackage{hyperref}
\usepackage{tkz-euclide,tkz-fct}
\usetkzobj{all}
\usepackage[titlenumbered,linesnumbered,ruled,noend,algosection]{algorithm2e}
\newcommand{\coll}{\operatornamewithlimits{\textsc{collinearPoints}}}

\newtheorem{thm}{Theorem}[section]
\newtheorem{cor}[thm]{Corollary}
\newtheorem{lem}[thm]{Lemma}

\numberwithin{equation}{section}

\begin{document}
\title{A Practical Algorithm for Enumerating Collinear Points}
\author{Ali Gholami Rudi\textsuperscript{1}}
\address{\textsuperscript{1}Faculty of
Electrical and Computer Engineering, Babol
Noshirvani University of Technology, Babol, Iran}
\email{gholamirudi@nit.ac.ir}
\author{Raimi Ayinde Rufai\textsuperscript{2}}
\address{\textsuperscript{2}SAP Labs, 111 rue Duke, Suite 9000, Montreal QC H3C 2M1, Canada}
\email{raimi.rufai@sap.com}

\maketitle

\begin{abstract}
This paper studies the problem of enumerating all maximal collinear subsets of
size at least three in a given set of $n$ points.
An algorithm for this problem, besides
solving degeneracy testing and the exact fitting problem,
can also help with other problems, such as point line cover and general position subset selection.
The classic \emph{topological sweeping} algorithm of Edelsbrunner and Guibas can find these subsets in
$O(n^2)$ time in the dual plane.  We present an alternative algorithm that, although asymptotically slower than their algorithm in the worst case, is simpler to implement and more amenable to parallelization.
If the input points are decomposed into $m$ convex polygons, our algorithm has time complexity $O(n^2 \log m)$ and space complexity $O(n)$. Our algorithm can be parallelized on the CREW PRAM with time complexity $O(n \log m)$ using $n$ processors.
\end{abstract}

\section{Introduction}
We study the problem of finding all maximal collinear subsets
of size at least three in a given set of points in the plane. In this paper, we assume the real RAM model of computation, where real arithmetic operations and comparison of reals take constant time and the floor function ($\lfloor a \rfloor$) is not allowed \cite{preparata2012computational}.

Our main motivation is that some of the algorithms for problems like point line cover (i.e.,
covering a set of points with the minimum number of lines)
\cite{kratsch16} or general position subset selection (finding the largest
subset of points in general position) \cite{froese16} need
to identify maximal collinear subsets as a first step.
A special case of this problem is the well-known degeneracy
testing problem (i.e.~testing whether any three of the points are
collinear).  The fastest known algorithm for this problem has time
complexity $O(n^2)$ \cite{edelsbrunner89}, which seems best possible, 
based on either the number of sidedness queries \cite{erickson96} or
3-{\small SUM} reduction \cite{gajentaan95}.  Another special case
is the exact fitting problem, which tries to find the line
that covers the most points \cite{guibas96}.
Finding maximal collinear subsets trivially solves these two problems.

As with degeneracy testing, the best known algorithm for finding
collinear subsets uses topological sweeping with time complexity
$O(n^2)$ and space complexity $O(n)$.

In this paper, we present an alternative algorithm based on cyclically
sorting the input points. The algorithm runs in time $O(n^2 \log m)$ and space $O(n)$, if the input points can be decomposed into $m$ convex polygons.  Our algorithm is much easier to implement than those depending on arrangements, and should be almost as fast in practice.  The techniques we use for sorting the points
cyclically may be helpful in other algorithms, such as those for
computing visibility graphs. Another advantage of our algorithm is that
it can be executed in parallel on the CREW PRAM in time $O(n \log m)$ and
space $O(nm)$ using $n$ processors.

The rest of this paper is organized as follows. We summarize related work in Section \ref{sec:related:work} and present the basis of
our algorithm in Section \ref{smain} with time complexity $O(n^2 \log n)$.
In Section \ref{spoly}, we improve its time complexity to $O(n^2 \log m)$,
if input points are decomposed into $m$ convex polygons. We present
our parallel algorithm in Section \ref{sect:parallel:algo} and end with some concluding remarks in Section \ref{sect:conclusion}.

\section{Related Work}
\label{sec:related:work}
The problem of finding collinear subsets can be restated in the dual
plane. Each point in the input is mapped to a line in the dual plane.
A subset of these lines intersect each other at a common point, if their
corresponding points in the original plane are collinear.
Therefore, the
problem of finding maximal sets of collinear points is equivalent
to finding the set of lines that intersect each other at each
intersection in the dual plane.
Intersecting lines can be identified using plane sweeping
with time complexity $O(n^2 \log n)$ \cite{bentley79}.
An alternative for identifying intersecting lines is using an
arrangement of lines, i.e.~a partition of the plane induced by
the set of lines into vertices, edges, and faces.
An arrangement has space complexity $O(n^2)$ and can be constructed
with time complexity $O(n^2)$ \cite{berg08}.

The well-known \emph{topological sweeping} algorithm of Edelsbrunner
and Guibas \cite{edelsbrunner89} can sweep arrangements without
constructing them, thus reducing its space complexity to $O(n)$.
The time complexity matches the lower-bound for the time
complexity of the problem of testing degeneracy and is thus the best
possible.  The algorithm, however, is rather difficult to implement,
especially since input points are not in general position (see for
instance \cite{rafalin02}). Note that reporting line intersections in
the dual plane is not enough for finding collinear
points.  These intersections should also be ordered so that all lines
that intersect can be reported at once efficiently, making
intersection reporting algorithms (for instance \cite{balaban95})
inefficient for this problem.

Because topological sweeping is inherently sequential and thus unsuitable for parallel execution, parallel
algorithms for sweeping or constructing arrangements have appeared in the literature. Anderson et al.~\cite{anderson96} presented a CREW PRAM algorithm for constructing arrangements in $O(n \log^{\ast} n)$ time with $O(n)$ processors and Goodrich \cite{goodrich91} presented an algorithm (also on the CREW PRAM) with the same goal
with time complexity $O(n \log n)$ with $O(n)$ processors.
Since these algorithms construct the arrangement, they have
space complexity $O(n^2)$, which is more than the $O(nm)$ space complexity of our parallel algorithm. Our algorithm also beats Goodrich's in time complexity, besides being easier to code since it uses only simple data structures.

Other methods have been presented to partition the arrangement into
smaller regions with fewer points and use the sequential algorithm for
sweeping these regions in parallel (e.g.~\cite{agarwal90} and \cite{hagerup90}).
They are usually based on the assumption that the points are in general position. Such  partitions yield poor performance, when this assumption is violated as it is the case in the problem of finding collinear points.
Similar parallel sweeping methods have been presented for specific applications
(e.g., rectangle intersection \cite{khlopotine13} and hidden surface elimination \cite{goodrich96}),
in which there are fewer events than the number of intersections
in the worst case (for finding collinear subsets $\Theta(n^2)$).
More recently, McKenny et al.~presented a plane sweep algorithm that divides
the plane into vertical slabs perpendicular to the sweep line, which has time
complexity $O(n \log n)$ with $n$ processors for $\Theta(n^2)$ intersections \cite{mckenney17}, whose cost is more than our algorithm.

\section{The Base Algorithm}
\label{smain}
In what follows, we describe our main algorithm for enumerating
maximal collinear subsets of a set of points.  It assumes an arbitrary ordering
$\sigma$ on the input points $P$. It processes each point $p$ (in the order specified by $\sigma$) and finds and reports all maximal sets of points collinear with $p$ that have not already been found in a previous iteration.
\begin{algorithm}
    \caption{$\coll(\sigma)$}
    \label{alg:BuildTree}

    \ForEach{$p$ \textbf{\emph{in}} $\sigma$}{
        Sort all points cyclically around $p$ in counterclockwise direction to obtain the sequence $O_p$. In other words, for each point $q$ in $P$,  let $a_p (q)$ be the counterclockwise angle between the horizontal half-line starting from $p$ rightward and the segment $pq$; $O_p$ is the result of sorting $P$ according to $a_p$. \label{stepi}\
        
        Decompose $O_p$ into sequences $U$ and $D$ while preserving the order of the points, such that $q \in U$ if $a_p (q) < \pi$, and $q \in D$, otherwise.  Let $\bar{a}_p (q) = a_p (q)$ for every point $q \in U$ and $\bar{a}_p (q) = a_p (q) - \pi$ for every point $q \in D$.  \label{stepii} \
        
        Merge $U$ and $D$ by keeping the points ordered by $\bar{a}_p$ to obtain $M_p$. \label{stepii-merge} \
        
        Find maximal consecutive collinear points in $M_p$. This requires a linear scan through $M_p$. \label{stepiii}\
        
        Report the sets found in step \ref{stepiii} that do not contain any point that precedes $p$ in $\sigma$. \label{stepiv}\
        
    } 
\end{algorithm}

It is not difficult to see that except for step \ref{stepi} of the algorithm,
each step has time complexity $O(n)$.  Sorting the points
in step \ref{stepi} can be performed in $O(n \log n)$.
Note that the points in sequences $U$ and $D$ remain sorted
according to both $a_p$ and $\bar{a}_p$ in step \ref{stepii} and
can be merged in $O(n)$.  Since these steps are repeated for
every point, the total time complexity of the algorithm is
$O(n^2 \log n)$ and its space complexity is $O(n)$.
Theorem \ref{tmain} shows its correctness.

\begin{thm}\label{tmain}
The algorithm reports every maximal set of collinear points exactly once.
\end{thm}
\begin{proof}
Let $S$ be a set of maximal
collinear points in $P$ and let $p$ be the first point of $S$ in $\sigma$.
When $P$ is sorted cyclically around $p$,
all points of $S$ appear contiguously either in
$U$ or in $D$, and when these sequences are merged, in $M_p$.
Thus, they are identified in step \ref{stepiii} and reported in step \ref{stepiv},
since $p$ was the first point of $S$ in $\sigma$.
For every other point $q$ of $S$, although the points of $S$
follow each other in $M_q$, they will not be reported in
step \ref{stepiv}, since $p$ is before $q$ in $\sigma$.
Note that the algorithm never outputs non-collinear points,
thanks to step \ref{stepiii}.
\end{proof}

In what follows, we try to improve the time complexity
of step \ref{stepi} (i.e.~sorting the points to obtain $O_p$).  A
simple heuristic, which we shall not pursue here, is to adjust $\sigma$ so that consecutive points
are as close as possible, so that fewer alterations are made
to $O_p$. This would help if in
step \ref{stepi}, the partially sorted sequence for the previous point is
sorted for the current point (some sorting algorithms are much
faster for sorting partially-sorted sequences).  More formally, 
let $p$ and $q$ be two consecutive points in $\sigma$.  A pair of points,
$s$ and $t$, is \emph{inverted} if $s$ appears before $t$ in $O_p$ but
after $t$ in $O_q$.  This can happen only if the line passing
through $s$ and $t$ intersects the segment from $p$ to $q$.
While moving a point $q'$ from $p$ to $q$, it can be shown that
each such intersection swaps two adjacent points
in $O_{q'}$.  Thus, the goal is to minimize the number of such
intersections or inverted pairs.

In the next section we use another method for the same goal:
finding sets of points, whose order does not change substantially in
$O_p$ for different points $p$ in $P$.

\section{Using Convex Polygon Decompositions}
\label{spoly}
We now concentrate on improving the complexity of computing $O_p$
for every point $p$ in $P$.  For that, we try to identify sets of
points in $P$ that appear in the same order in $O_p$ for every $p$
in $P$.
As we show in Lemma \ref{lpoly}, convex polygons have this property
to a large extent.
A \emph{tangent} $l_{pq}$ from a point $p$ to a convex polygon $P$ is a line that passes through $p$
and another point $q$ of $P$ such that all of $P$ lies to one side of $l_{pq}$; $q$ is called the \emph{tangent point} from $p$ to polygon $P$. If the tangent passes through more than one point of $P$, the closest to $p$ is considered the tangent point.  Tangent points can be identified via a linear scan (or binary search) through the points of $P$.

\begin{lem}
\label{lpoly}
For any point $p$ and any convex polygon $P = \left( p_0 , p_1 ,..., p_{k-1} \right)$,
$P$ can be decomposed into at most two sequences such that the points in
any of these sequences are cyclically sorted (either clockwise or counterclockwise) from $p$.
\end{lem}
\begin{proof}
If $p$ is inside $P$, every point of $P$
appears in the same order when cyclically ordered from $p$; thus,
the whole of $P$ is a sorted sequence.
Otherwise, let $p_i$ and $p_j$ be tangent points from $p$
to polygon $P$.  Then,
the sequences $\left( p_i , p_{i+1} ,..., p_{j-1} \right)$ and
$\left( p_j , p_{j+2} ,..., p_{i-1} \right)$ are cyclically sorted
from $p$ (indices are modulo $k$).
\end{proof}

\begin{thm}
\label{tmerg}
If the set of points $P$ can be decomposed into $m$ convex polygons,
the algorithm introduced in Section \ref{smain} for identifying
collinear points can
be improved to achieve the time complexity of $O(n^2 \log m)$.
\end{thm}
\begin{proof}
We modify step \ref{stepi} of the algorithm for calculating $O_p$.
Based on Lemma \ref{lpoly}, we can obtain $k$ sequences ($k \le 2m$),
all of which are sorted cyclically around $p$.  This can be done in $O(n)$
(to identify the tangent points of the polygons and to create the new sequences).
We have to merge these sorted sequences to obtain $O_p$.
This can be done in $O(n \log m)$ using a heap priority queue:
initially the first points of the sorted sequences are inserted into the heap
with time complexity $O(m \log m)$.  Then, the point $q$ with the smallest value
of $a_p (q)$ is extracted from the heap and the point with the
next smallest $a_p$ in $q$'s sequence is inserted into the heap.
This process is repeated $n$ times to obtain $O_p$.  Since extracting
the value from and inserting a new value into the heap can be done in $O(\log m)$,
$O_p$ can be constructed with time complexity $O(n \log m + m \log m)$,
which is equivalent to $O(n \log m)$, since $m < n$.
\end{proof}

One method for decomposing a set of points into convex polygons is
convex hull peeling, i.e., repeatedly extracting convex hulls from
the set (for applications and a survey, the reader may consult \cite{rufai15}).
The number of resulting convex hulls is sometimes called convex
hull peeling depth.  Common convex hull algorithms
can be used for convex hull peeling by repeatedly extracting convex hulls
in $O(mn \log n)$, if the convex hull peeling depth is $m$ (this
seems adequate for our purpose; there are
faster algorithms however \cite{chazelle85}).  This yields the
following corollary.
Note that most convex hull algorithms
can be slightly modified to allow collinear points in the boundary of
the hull.

\begin{cor}
All maximal collinear subsets of size at least three in a set of $n$
points can be identified in $O(n^2 \log m)$, where $m$ is the convex
hull peeling depth of the points.
\end{cor}

Parts (a)--(d) of Figure \ref{fig:layers:delete:case3} demonstrate the steps
of the algorithm described in this section for decomposing the set
of input points into convex polygons (part (b)), splitting each
polygon into at most two cyclically sorted sequences (part (c)),
and merging these sequences to obtain $O_q$ (part (d); the
arrows show the order of processing the points in the sequences).
The points on each of the lines $L_1$, $L_2$, and $L_3$ should
be contiguous in $M_p$.

\begin{figure}[htbp]
  \centering
    \includegraphics[scale=0.75]{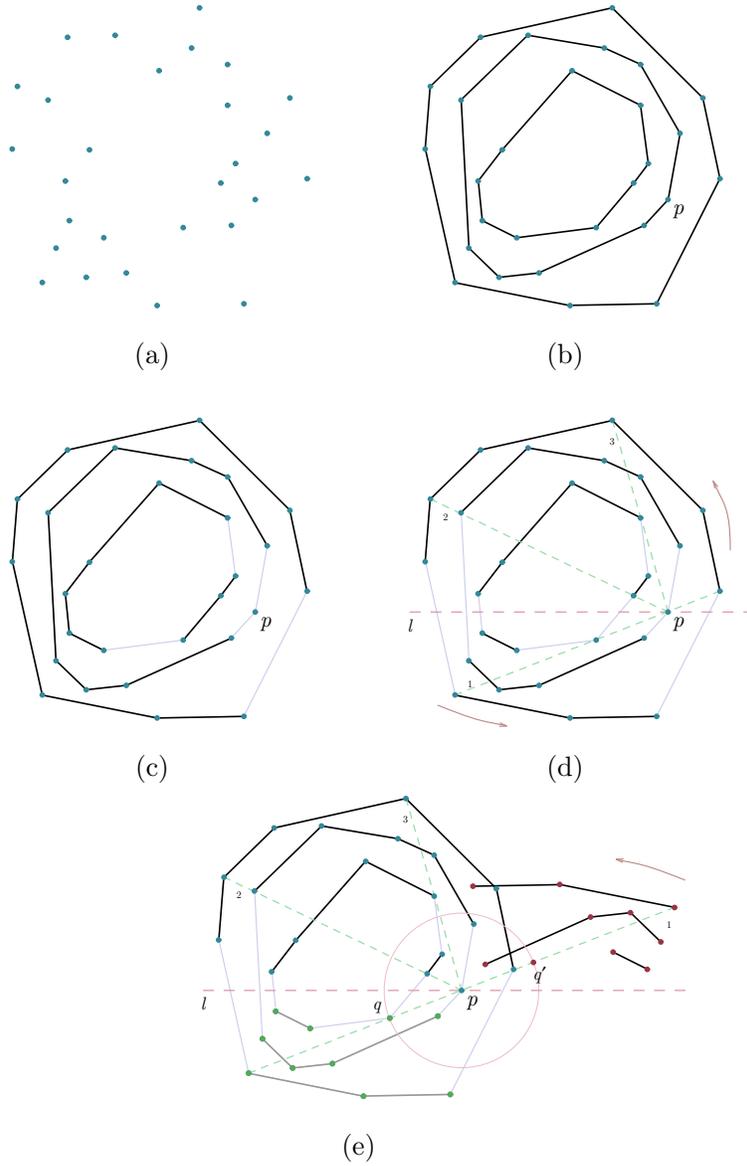}
  \caption{\textbf{(a)--(b)} An input point set $P$ is decomposed into convex layers. The algorithm is currently processing point $p$. \textbf{(c)} The convex layers are decomposed into sorted subsequences, bookended by points tangent to $p$. \textbf{(d)} Subsequences are split up and added to $U$ or $D$ depending on whether they lie above or below the horizontal line $l$ through $p$. (\textbf{e}) Points appearing in $D$ are then mapped to their antipodes and merged with $U$ to obtain $M_p$. Collinear points appear contiguous in $M_p$.  }
  \label{fig:layers:delete:case3}
\end{figure}

\section{Parallel Algorithm}
\label{sect:parallel:algo}
To obtain a parallel algorithm for finding sets of collinear points,
both the algorithm presented in Section \ref{smain} and convex
polygon decomposition should be performed in parallel, as
Theorem \ref{tpar} shows. We shall use the following lemma in the proof of Theorem \ref{tpar}.

\begin{lem}
\label{lcontra}
If $m < n$, then $m \log n = O(n \log m)$.
\end{lem}

\begin{proof}
Suppose for the sake of contradiction that the converse is true, i.e.~if
$m < n$, then $m \log n > c n \log m$ for some $c \ge 1$ and $n \ge 3$.
This implies that $ m/\log m > cn/\log n$, a contradiction,
since the function $f(x) = x/\log x$ for $x \ge 3$ is
monotonically increasing.
\end{proof}

We are now ready to prove Theorem \ref{tpar}.

\begin{thm}
\label{tpar}
With $O(n)$ processors, it is possible to identify all maximal collinear
subsets of a set of $n$ points on the CREW PRAM with time complexity $O(n\log m)$
and space complexity $O(nm)$,
where $m$ is the convex hull peeling depth of the points.
\end{thm}
\begin{proof}
The parallel algorithm first decomposes the input points into convex polygons
and then enumerates collinear subsets of the points.  We discuss these two
steps separately as follows.

For decomposing points into convex polygons, we use the parallel
algorithm presented by Akl \cite{akl84}: it computes the convex hull
of a set of $n$ points in $O(\log n)$ time with $O(n)$ processors.
Akl's algorithm assumes input points to be sorted by their $x$-coordinates;
this can be done on the CREW PRAM in $O(\log n)$ with $O(n)$
processors \cite{cole88}.
The convex hull algorithm
should be repeated $m$ times, yielding the time complexity of $O(m\log n)$.
Thus, the total time complexity of the algorithm is $O(n \log m + m \log n)$,
which is equivalent to $O(n \log m)$.
This follows trivially if $m = \Theta (n)$.
However, if $m = o(n)$ it follows from Lemma \ref{lcontra}.

Given that there is no dependency between different iterations of
the algorithm presented in Section \ref{smain},
it can be parallelized by distributing the points among the processors,
each of which requires $O(n\log m)$ to finish its task with space
complexity $O(n)$.
Therefore, it remains to improve the overall space
complexity of the parallel algorithm to $O(nm)$.

The $O(n)$ space complexity per processor is due to storing the
sorted and merged sequences of points ($O_p$ and $M_p$, as defined in Section \ref{smain}),
storing sequences of points obtained after splitting convex polygons,
and storing consecutive collinear points in step \ref{stepiii} of the algorithm
presented in Section \ref{smain}.  We reduce the space complexity of
each processor to $O(m)$ (plus $O(n)$ for storing convex polygons,
which is shared among the processors).
Instead of storing all elements of a subsequence of a sequence
of points, we store two pointers to indicate its start and end
positions.  Therefore, we can store the subsequences resulting from
splitting the polygons in $O(m)$ (there are at most $2m$ such subsequences).

We now modify the base algorithm not to store $O_p$ and $M_p$.
We use some of the symbols defined in Section \ref{smain}: $a_p (q)$
and $\bar{a}_p (q)$ for point $q$.
Let $S$ be the set of cyclically sorted sequences obtained by splitting
convex polygons, as explained in Theorem \ref{tmerg}.
We split each sequence $s$ in $S$ into at most two sequences to
obtain the set $S'$ (this still requires $O(m)$ words of memory,
as the resulting subsequences are also subsequences of the convex polygons):
Let $s_1$ be the subsequence of points
in $s$, for which $a_p (q) < \pi$ and $s_2$ be the rest of
the sequence.
We insert the first point in each sequence in $S'$ to a priority
heap, as described in the proof of Theorem \ref{tmerg}.
Then, instead of merging these subsequences to obtain $M_p$, we
detect collinear points (steps \ref{stepiii} and \ref{stepiv} of the base algorithm)
as we extract each point $q$ from the heap based on their $\bar{a}_p(q)$.
This is demonstrated in part (e) of Figure \ref{fig:layers:delete:case3}:
for points like $q$ below $p$, $\bar{a}_p (q) = a_p(q) - \pi$.  Therefore,
the algorithm virtually considers their antipodal points when extracting
the minimum from the heap; these antipodal points are coloured red
in the figure (for instance, $q'$ for $q$).
The arrow shows the order of inserting points into the heap for each
sequence.  The three points in $L_1$, including $p$ and $q'$, are
extracted successively from the heap, as required in step \ref{stepiv} of the
base algorithm.
Given there are $O(m)$ items in the heap at any moment,
its space complexity is also $O(m)$.
\end{proof}

\section{Conclusion}
\label{sect:conclusion}
We have presented a simple sequential algorithm for finding collinear
$\mathcal{R}^2$ points in the real RAM model that initially ran in
time $O(n^2 \log n)$ and linear space. We then improved its running
time to $O(n^2 \log m)$ by first decomposing the points into $m$ convex
layers. Decomposition into convex layers by repeatedly invoking an
optimal convex hull algorithm takes $O(mn \log n)$ time. Note that
this is also $O(n^2 \log m)$ by Lemma \ref{lcontra}. By taking
advantage of the ordering in the constituent convex polygons, we were
able to avoid the sorting step and reduce the time
complexity to $O(n^2 \log m)$ using linear space. Finally, we showed a
parallel version of the algorithm on the CREW PRAM that runs in
$O(n \log m)$ time using $O(mn)$ space and $O(n)$ processors.

While we have used the convex layers decomposition to obtain the convex
polygons used in our algorithm, we could have used any other
decomposition. Particularly, it would be interesting to explore the
space of convex decompositions that minimize $m$. We have restricted
our attention to point sets in the plane, but hope to explore
generalizations to higher dimensions in future work. In Section \ref{smain}, we had hinted at the idea of finding orderings that minimize inversions between  iterations of the base algorithm. It would be interesting to further explore this idea in a future work.   We have also
worked within the confines of the real RAM model. It would be
interesting to explore algorithms for the same problem in alternative
models of computation.

\bibliographystyle{unsrt}

\end{document}